\documentclass[12pt]{article}
\usepackage{moreverb}
\usepackage{epsfig}
\usepackage{graphicx}
\usepackage{subfigure}

\newtheorem{theorem}{Theorem}

\newtheorem{lemma}{Lemma}

\newtheorem{definition}{Definition}

\newtheorem{proof}{Proof}
\newtheorem{remark}{Remark}
\usepackage{amsmath}
\usepackage{amssymb}

\usepackage[dvips,colorlinks,bookmarksopen,bookmarksnumbered,citecolor=red,urlcolor=red]{hyperref}

\begin{document}



\title{Estimation and Control over Cognitive Radio Channels with Distributed and Dynamic Spectral Activity}

\author{Xiao Ma, Seddik M. Djouadi, Husheng Li, and Teja Kuruganti, 
\thanks{X. Ma is with Western Digital, Irvine, CA, 92612 USA. Email: \{xma8623@gmail.com\}. S.M. Djouadi and H. Li are with the Department of Electrical Engineering and Computer Science, University of Tennessee, Knoxville, TN 37996 USA. E-mail: \{mdjouadi@utk.edu\}, and \{husheng@eecs.utk.edu\}. T. Kuruganti is with Oak Ridge National Lab, Oak Ridge, TN, 37831 USA. e-mail: \{kurugantipv@ornl.gov\}}}



\maketitle

\begin{abstract}
Since its first inception by Joseph Mitola III in 1998 cognitive radio (CR) systems have seen an explosion of papers in the communication community. However, the interaction of CR and control has remained vastly unexplored. In fact, when combined with control theory CR may pave the way for new and exciting control and communication applications. In this paper, the control and estimation problem via the well known two switch model which represents a CR link is considered. In particular, The optimal linear estimator subject to a CR link between the sensor and the estimator is derived. Furthermore, it is shown that in the Linear Quadratic Gaussian (LQG) Control law for a closed-loop system over double CR links is not linear in the state estimate. Consequently, the separation principle is shown to be violated. Several conditions of stochastic stability are also discussed. Illustrative numerical examples are provided to show the effectiveness of the results.
\end{abstract}

{\bf Keywords}: Cognitive radio; Estimation; Control; Two switch model.



\section{Introduction}
\label{sec:introduction}
 Rapid advances in communication and networking extends the areas of traditional engineering sciences. These wireless techniques are widespread to ease applications in different applications. However, the wide use of various technologies, such as radio, satellite, and phone service, also increase the need of bandwidth used for transmission. Most of the current bandwidth spectrum has been licensed to different users to ensure the coexistence of diverse wireless users \cite{IEEECM:Srinivasa}. Thus an important question: \emph{How can communication bandwidth be saved without affecting the performance significantly}?

The Federal Communications Commission's (FCC) frequency allocation chart \cite{FCC:NTIA} shows that although the majority of frequency bandwidth has been assigned to different users, large portions of it are not frequently used \cite{FCC:FCC}. To increase spectrum use, \textbf{cognitive radio architecture} \cite{KTH:Mitola} \cite{CN:Akyildiz}\cite{TVT:Yang} is proposed to sense available spectrum, search for unutilized spectrum, and communicate over the latter with minimal disturbance to primary users (with license). Each secondary user (without license) is able to sense the licensed spectrum band and detect unused spectrum holes. If a frequency channel is not being used by primary users, secondary users can access it for communication. Due to the sparse activities of primary users, cognitive radio can provide a large amount of spectrum for communications. Thus cognitive radio answers the question above as bandwidth is exploited efficiently resulting in money savings for transmission.

An interesting application of CR is in control engineering, for example in the smart grid \cite{IEEECST:xi}, where power systems rely on control algorithms for power regulation and management. Besides, applications of state estimation over CR are presented in \cite{TWE:xiao}. However, there are new issues that need to be addressed. For instance, when a user wishes to do remote control without having any authorized bandwidth or without enough funds to purchase it, CR can be employed to help the user reach his target. However, CR suffers from interruptions from primary users since secondary users must leave licensed channels when the former emerge. Hence, the cognitive radio-based communication link may not be reliable causing significant impact on control and state estimation, since sensor observations may not reach the controller in a timely fashion.

Modern control theory has been increasingly concerned with networks, communication channels, and remote control technology. Much research has been performed in the area of control and estimation over communication links under packet losses \cite{IEEETIT:Hadidi}$\sim$\cite{NOLC:Imer}. In this paper, we study the estimation and control problems through CR links with distributed and dynamic spectral activity using the important two-switch model proposed in \cite{IEEECM:Srinivasa}. This model includes two Bernoulli random variables that depend on each other, and represents primary users (PUs) interruption to secondary users' (SUs) transmitter and receiver, respectively. Compared with the existing work on estimation over lossy networks, our paper studies a mathematical model which has two lossy indicators ($s_t$ and $s_r$), which is different from all existing works that only consider one lossy indicator. The two lossy indicators are different with each other, e.g., $s_t$ is unknown to the receiver while $s_r$ is known. Moreover, $s_t$ also depends on $s_r$. Thus, from the mathematical point of view, this two-switch model is more general and includes the single lossy indicator case (e.g., \cite{IEEETAC:Sinopoli}, \cite{IEEETIT:Nahi}, \cite{IEEE:Schenato}, \cite{SCL:Gupta}, and \cite{Auto:Huang}, etc.). The algorithm design would take the information of both $s_r$ and $\mathbb{P}(s_t)$ since $s_t$ is unknown to the receiver.

The contributions of this paper are summarized as follows:
\begin{itemize}
\item The optimal linear estimator over a single CR link between the sensor and the estimator is derived;
\item Estimation and control of closed-loop system over CR links are addressed. The controller is shown to be a nonlinear function of the state estimate and depends on the error covariance.
\item It is shown that the principle of separation does not hold.
\item A linear feedback controller is employed in the closed-loop system and several stochastic stability conditions are derived.
\end{itemize}
A short preliminary version of some of the results appeared in \cite{ACC:xiao}.
\\
The paper is organized as follows: in section \ref{sec:twoSwitch}, the two switch model for a CR system is discussed. In Section \ref{sec:LOECR}, the linear optimal estimator for the system is derived. In Section \ref{sec:ctrE}, estimation and control of the closed-loop system over CR links is addressed. The conclusion and further works are discussed in section \ref{sec:con}.

\section{The Two-Switch Model} \label{sec:twoSwitch}

In this section, we introduce the two switch model used to model CR systems throughout this paper.

The general idea of CR system can be interpreted by Fig. \ref{fig:bandwidth}. Assume there are \emph{\textbf{N}} independent licensed channels that can be sensed named as $f_1, \; f_2,\; \cdots,\;f_N$; each channel is divided into parts by vertical lines and each part represents that channel in one time slot; the hatched slot represents that the channel is utilized by PUs and the SUs can not use it while the blank square means that it is free to be used by other users \cite{ACC:xiao}.

\begin{figure}
 \begin{center}
  \includegraphics[width=8cm]{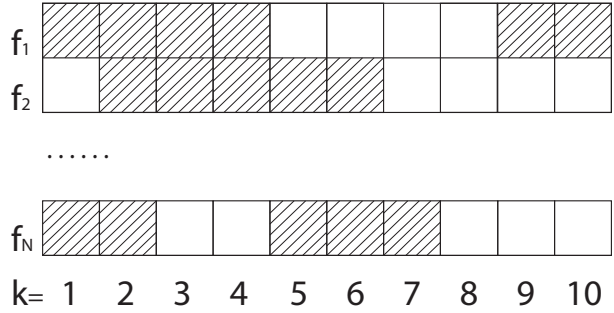}\\
   \end{center}
  \caption{Channels in CR system.}\label{fig:bandwidth}
\end{figure}

In CR systems, PUs represent the spectrum as they pay or as they are assigned to it. SUs take advantage of inactivity periods of PUs to transmit information through the available channels. SUs have to avoid transmitting to minimize interference with PUs. The two switch model considered in this paper is well known and was proposed in the communication community \cite{IEEECM:Srinivasa}.

First, consider the CR link shown in Fig. \ref{fig:bandwidth}. We assume one secondary transmitter (ST) and one secondary receiver (SR) in the presence of several PUs, for e.g., 3 PUs A, B, and C (for convenience of illustration). Two circles represent the sensing areas where the ST and SR can detect the activities of PUs. In Fig. \ref{fig:twoswitch}, for example, the ST can only sense whether A or B are active, and then reports that spectrum as available for transmission when both A and B are inactive. Similarly, the SR does the same for B and C. As a result, ST and SR may detect unused spectrum at different times.



The conceptual model in Fig. \ref{fig:twoswitch} produces the two-switch mathematical model shown in Fig. \ref{fig:twoMath}, where $s_t$ and $s_r$ denote the sensing variables of ST and SR. Let $s_t=0$ if ST senses active PUs and $s_t=1$ if no active PUs. $s_r=0$ if SR senses active PUs and $s_r=1$ if no active PUs. We also assume that PUs are independent with each other \cite{IEEECM:Srinivasa}. The circles in Fig. \ref{fig:twoswitch} represent the sensing regions.

\begin{figure}
 \begin{center}
  \includegraphics[width=8cm]{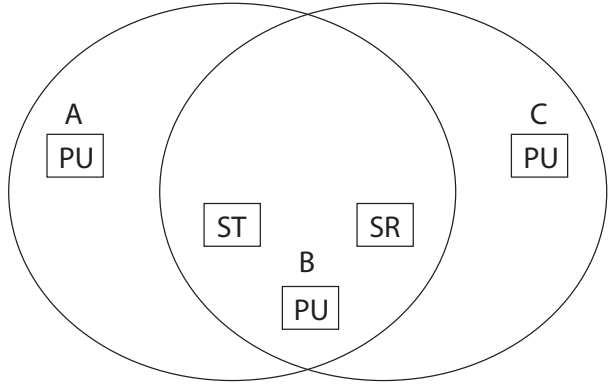}\\
   \end{center}
  \caption{Conceptual model of a cognitive radio system with ST and SR \cite{IEEECM:Srinivasa}.}\label{fig:twoswitch}
\end{figure}
\begin{figure}
 \begin{center}
  \includegraphics[width=8cm]{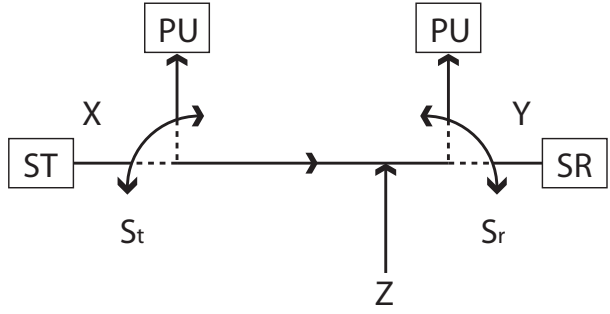}\\
   \end{center}
  \caption{Mathematical model of two-switch model \cite{IEEECM:Srinivasa}.}\label{fig:twoMath}
\end{figure}

The switch state $s_t$ is known only to the transmitter, while $s_r$ is known only to the receiver. Correlation exists between these two switch states as can be seen from Fig. \ref{fig:twoswitch}. They both depend on the PUs that exist in the intersecting sensing regions. The mathematical model can be written as:
\begin{equation}
Y = {s_r}({s_t}X + Z)
\end{equation}
where $Y$ is the received signal, and $X$ and $Z$ are the transmitted signal and noise, respectively \cite{IEEECM:Srinivasa}.

\section{Optimal Linear Estimator Via Cognitive Radio} \label{sec:LOECR}

In this section, we derive the linear optimal estimator for a single CR system as shown in Fig. \ref{fig:estimation}. Due to the independence of each channel, without loss of generality, only one channel in the CR system for is considered for convenience. The case of multiple channels can be extended from the single channel case easily and will be discussed briefly.

\begin{figure}
 \begin{center}
  \includegraphics[width=8cm]{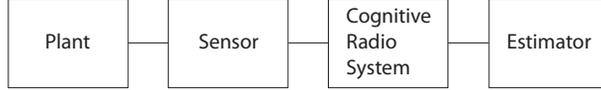}\\
   \end{center}
  \caption{Estimation over cognitive radio system.}\label{fig:estimation}
\end{figure}

\subsection{Problem Formulation}\label{sec:pf}

First, we consider estimation over a single CR link between the sensor and the estimator as shown in Fig. \ref{fig:estimation}. Let us denote by $(\Omega ,\Gamma ,P)$ the probability space induced by the following linear stochastic discrete time-invariant system:
\begin{eqnarray}
{x_k} &=& A{x_{k - 1}} + {\upsilon _k}\nonumber\\
{y_k} &=& s_r^k(s_t^kC{x_k} + {\omega _k})
\label{eq:CREstimationPro}
\end{eqnarray}
where ${x_k} \in {\mathbb{R} ^d}$ is the state at $k$, ${y_k} \in {\mathbb{R} ^l}$ is the observation received, $x_0$ is the initial value of the processes ${\{ {x_k}\} _{k \in \mathbb{N}}}$. ${\upsilon _k} \in {\mathbb{R} ^d}$ and ${\omega _k} \in {\mathbb{R} ^l}$ are independent Gaussian white sequences with zero mean and positive definite covariance matrices $V$ and $W$. The matrices $A$ and $C$ are respectively the state and measurement matrices of the system with appropriate dimensions. $s_{t}^{k}$ and $s_{r}^{k}$ are switching random variables of ST and SR at $k$, respectively. A CR system is located between the sensor and the estimator. Assume that $\{s_{t}^{k}\}_k$ is an independent and identically distributed (i.i.d) Bernoulli random processes with the probability $\mathbb{P}\{ s_t^k = 1\}  = \gamma $, and $\{s_{r}^{k}\}$ is another i.i.d Bernoulli random process with the probability $\mathbb{P}\{ s_r^k = 1\}  = q$. Note that the two Bernoulli variables may depend on each other due to the intersection of the sensing regions. They are assumed to be independent of the state variables and noise signals. The switching variable $s_{t}^{k}$ is not known while $s_{r}^{k}$ is known at the receiver. Denote by ${\{ {I_k}\} _{k \in N}}$ the complete filtration ($\sigma $-algebra) generated by observations $\{ {y_1},\cdots,{y_k},s_r^1,\cdots,s_r^k\} $. Moreover, let $(A,V^{\frac{1}{2}})$ be controllable, $(A,C)$ and $(A,W^{\frac{1}{2}})$ observable \cite{John:Simon,Athe:Berrsekas}.

The optimal estimation problem can be posed as the minimization of the following cost function \cite{John:Simon}:
\begin{equation}
{J_k} = \mathbb{E}\{ {({x_k} - {\hat x _{k|k}})^T}({x_k} - {\hat x _{k|k}})|I_k\}
\label{eq:CREstimationCost}
\end{equation}
with respect to the state estimate ${\hat x _{k|k}}$. Computing this estimate is the subject of the next section.

\subsection{Linear Optimal Estimator}\label{sec:LOE}
In this section, we derive the linear optimal estimator by assuming that the state estimate is a linear function of measurements.

The optimal state estimate is ${\hat x _{k|k}} = \mathbb{E}\{ {x_k}|{I_{k}}\}$ which would be a nonlinear function of the measurement \cite{IEEESPL:xiao}. This requires exponentially increasing memory of computation cost \cite{IEEESPL:xiao}. Thus, we consider a linear formulation in this work. The linear optimal estimator could be derived by using a linear term of measurements in the estimate.

\begin{theorem}\label{thm:loe}
The linear estimator that minimizes the cost function (\ref{eq:CREstimationCost}) is given by:
\begin{eqnarray}
{\hat x _{k|k - 1}} &=& A{\hat x _{k - 1|k - 1}}\label{eq:xpri}\\
{P_{k|k - 1}} &=& A{P_{k - 1|k - 1}}{A^T} + V\label{eq:ppri}\\
{\hat x _{k|k}} &=& {\hat x  _{k|k - 1}} + {K_k}({y_k} - s_r^kpC{\hat x _{k|k - 1}})\label{eq:xpost}\\
{P_{k|k}} &=& {P_{k|k - 1}} - s_r^kp{K_k}C{P_{k|k - 1}}\label{eq:ppost}\\
{K_k} &=& {P_{k|k - 1}}p{C^T}{(p^2 C{P_{k|k - 1}}{C^T} + {W_k^{'}})^{ - 1}}\label{eq:filtergain}\\
{W_k^{'}} &=& \,\,\,W + (p - {p^2})C{X_k}{C^T}\\
{X_{k}} &=& A{X_{k-1}}{A^T} + V
\label{eq:loeCR}
\end{eqnarray}
where ${p} = \mathbb{P}(s_t^k = 1|s_r^k = 1)$ and ${X_1} = {x_1}x_1^T + {P_1}$.
\end{theorem}
\begin{proof}:
The predictable state, i.e., the estimate at time $k$ based on the observations up to time $k-1$ is:
\[{\hat x _{k|k - 1}} = \mathbb{E}\{ {x_k}|{I_{k - 1}}\}  = A{\hat x  _{k - 1|k - 1}}\]
\begin{eqnarray}
{P_{k|k - 1}} &=& \mathbb{E}\{ ({x_k} - {\hat x _{k|k - 1}}){({x_k} - {\hat x _{k|k - 1}})^T}|{I_{k - 1}}\}\nonumber\\
&=& A{P_{k - 1|k - 1}}{A^T} + V\nonumber
\end{eqnarray}
where ${\hat x _{k|k - 1}}$ is the a priori state estimate and $\hat x  _{k |k}$ is a posteriori state estimate at time $k$, ${P_{k|k - 1}}$ is the covariance of the estimation error of ${x_k} - {\hat x _{k|k - 1}}$;   $P_{k - 1|k - 1}$ is the covariance of the estimation error ${x_k} - {\hat x _{k-1|k - 1}}$.

We assume, as in the traditional Kalman Filter, that the state is linear in the innovation process (which is $({y_k} - s_r^ks_t^kC{\hat x _{k|k - 1}})$), that is $x_k={\hat x _{k|k - 1}} +{K_k}({y_k} - s_r^ks_t^kC{\hat x _{k|k - 1}})$, then we have:
\begin{eqnarray}
{\hat x _{k|k}} &=& \mathbb{E}\{x_k |I_k\}\nonumber\\
&=& {\hat x _{k|k - 1}} +\mathbb{E}\{ {K_k}({y_k} - s_r^ks_t^kC{\hat x _{k|k - 1}})|{I_k}\}\nonumber\\
\label{eq:PosterState}
\end{eqnarray}
where $K_k$ is the linear optimal estimator gain matrix at time $k$ and ${y_k} - s_r^ks_t^kC{\hat x _{k|k - 1}}$ is the innovation process.

In (\ref{eq:PosterState}) except for $s_r^ks_t^kC{\hat x _{k|k - 1}}$, the other terms do not depend on $s_{t}^{k}$ and $s_{r}^{k}$ \textbf{(note that $I_k=\{ {y_1},\cdots,{y_k},s_r^1,\cdots,s_r^k\}$, so $y_k$ is known)}. Thus, (\ref{eq:PosterState}) becomes:
\[\begin{array}{l}
 {\hat x _{k|k}} 
  = {\hat x _{k|k - 1}} + {K_k}({y_k} - s_r^k\mathbb{E}\{ s_t^k|s_r^k\} C{\hat x _{k|k - 1}})
 \end{array}\]
since $s_r^k$ is measurable with respect to $I_k$ and $s_t^k$ only depends on $s_r^k$.

Define ${\varepsilon _{x,k|k}}$ as the estimation error at $k$ with observations up to time $k$; ${\varepsilon _{x,k|k-1}}$ as the estimation error at $k$ with observations up to $k-1$. Then, the estimation error mean can be written as:
\begin{eqnarray}
 \mathbb{E}\{ {\varepsilon _{x,k|k}}|{I_k}\}  &=& \mathbb{E}\{ {x_k} - {\hat x _{k|k}}|{I_k}\}  \\
 &=& \mathbb{E}\{ {x_k} - {\hat x _{k|k - 1}} - {K_k}({y_k} - s_r^k{p_t}C{\hat x _{k|k - 1}})|{I_k}\}  \\
  &=& (I - s_r^k{K_k}{p_t}C)\mathbb{E}\{ {\varepsilon _{x,k|k - 1}}|{I_k}\}  - s_r^k{K_k}\mathbb{E}\{ \omega _k^{'}|{I_k}\}
\end{eqnarray}
where ${p_t} = :\mathbb{E}\{ s_t^k|s_r^k\} $; $\omega _k^{'}\,: = {\omega _k} + (s_t^k - {p_t})C{x_k}$ can be viewed as the new measurement noise. Then, by independence of the state, the noise and $s_t^k$, we have:
\[\begin{array}{l}
\mathbb{E}\{ \omega _k^{'}|{I_k}\}  = \mathbb{E}\{ {\omega _k} + (s_t^k - {p_t})C{x_k}|{I_k}\}  \\
= \mathbb{E}\{ {\omega _k}\}  + \mathbb{E}\{ s_t^k - {p_t}|{I_k}\} \mathbb{E}\{ C{x_k}|{I_k}\}  = 0
 \end{array}\]

Also we have $\mathbb{E}\{ \omega _k^{'}\upsilon _{_k}^T\}  = 0$. Then the estimation error covariance $P_{k|k}$ at time $k$ is:
\[\begin{array}{l}
 {P_{k|k}} = \mathbb{E}\{ {\varepsilon _{x,k|k}}{\varepsilon }_{x,k|k}^T|{I_k}\}  
  = (I - s_r^k{K_k}{p_t}C)\mathbb{E}\{ {\varepsilon _{x,k|k - 1}}{\varepsilon}_{x,k|k - 1}^T|{I_k}\} {(I - s_r^k{K_k}{p_t}C)^T} \\
  - s_r^k{K_k}\mathbb{E}\{ \omega _k^{'}\varepsilon _{x,k|k - 1}^T|{I_k}\}\times {(I - s_r^k{K_k}{p_t}C)^T} 
  - (I - s_r^k{K_k}{p_t}C)\mathbb{E}\{ {\varepsilon _{x,k|k - 1}}\omega _k^{{'}T}|{I_k}\} s_r^kK_k^T \\ + s_r^k{K_k}\mathbb{E}\{ \omega _k^{'}\omega _k^{{'}T}|{I_k}\} s_r^kK_k^T
 \end{array}\]

Note that ${\varepsilon _{x,k|k - 1}}$ is the estimation error at time $k$ before receiving the measurement, $\omega _k^{'}$ is combined with the measurement noise $\omega _k$ at time $k$. Thus $\omega _k^{'}$ is independent of ${\varepsilon _{x,k|k - 1}}$. Therefore, $\mathbb{E}\{ \omega _k^{'}\varepsilon _{x,k|k - 1}^T|{I_k}\}  = \mathbb{E}\{ {\varepsilon _{x,k|k - 1}}\omega _k^{{'}T}|{I_k}\}  = 0$ and we have
\begin{eqnarray}
{P_{k|k}} &=& (I - s_r^k{K_k}{p_t}C){P_{k|k - 1}}{(I - s_r^k{K_k}{p_t}C)^T} + s_r^k{K_k}{W_k^{'}}K_k^T
\label{eq:ErrorCov}
\end{eqnarray}
where ${W_k^{'}} = \mathbb{E}\{ \omega _k^{'}\omega _k^{{'}T}|{I_k}\}$ is the variance of ${w_k}^{'}$ and is determined by
\[\begin{array}{l}
 {W_k^{'}} = \,\mathbb{E}\{ \omega _k^{'}\omega _k^{'T}|{I_k}\} 
 = \mathbb{E}\{ ({\omega _k} + (s_t^k - {p_t})C{x_k}){({\omega _k} + (s_t^k - {p_t})C{x_k})^T}|{I_k}\}  \\
  = W + (\mathbb{E}\{ {(s_t^k)^2}|s_r^k\}  - {({p_t})^2})C{X_k}{C^T}
 \end{array}\]
where ${X_k} = \mathbb{E}\{ {x_k}x_k^T|I_k\}$.

Following \cite{Auto:Koning}, we obtain ${X_{k + 1}} = A{X_k}{A^T} + V$, and ${X_1} = {x_1}x_1^T + {P_1}$ to make $\{ {X_k}\} $ a known sequence.

The optimality criterion is to minimize the cost function ${J_k}$ at $k$. Note ${J_k} = Trace({P_{k|k}})$ \cite{John:Simon}. Differentiating $J_k$ with respect to (w.r.t) $K_k$ yields
\begin{equation}
 \frac{{\partial {J_k}}}{{\partial {K_k}}} 
= 2(I - s_r^k{K_k}{p_t}C){P_{k|k - 1}}( - s_r^k{p_t}{C^T}) + 2s_r^k{K_k}{W_k^{'}}
 \label{eq:MinCostCR}
\end{equation}

Letting (\ref{eq:MinCostCR}) be equal to 0, and solving for $K_k$ yields:
\[{K_k} = {P_{k|k - 1}}{p_t}{C^T}{({p_t}C{P_{k|k - 1}}{p_t}{C^T} + {W_k^{'}})^{ - 1}}\]
and plugging $K_k$ in (\ref{eq:ErrorCov}) gives:
\begin{equation}
{P_{k|k}} = {P_{k|k - 1}} - s_r^kp_t{K_k}C{P_{k|k - 1}}\label{eq:ppost2}
\end{equation}

Next $p_t$ is computed. As $s_t^k \in \{ 0,1\}$, ${p_t} = 1 \times \mathbb{P}(s_t^k = 1|s_r^k) + 0 \times \mathbb{P}(s_t^k = 0|s_r^k) = \mathbb{P}(s_t^k = 1|s_r^k)$, which includes two cases: $s_r^k=0$ and $s_r^k=1$. Note when $s_r^k = 0$, the receiver is inactive to avoid disturbing PUs, so $y_k=0$. Then, the second terms on the right hand side in both (\ref{eq:PosterState}) and (\ref{eq:ppost2}) vanishes, which means when $s_r^k = 0$, ${p_t}$ does not affect the values of the state estimate and error covariance; Thus, we only need to compute the case when $s_r^k = 1$, that is ${p_t} = \mathbb{P}(s_t^k = 1|s_r^k = 1)$ and (\ref{eq:PosterState}) can be represented as:
\[{\hat x _{k|k}} = {\hat x _{k|k - 1}} + {K_k}({y_k} - s_r^k\mathbb{P}(s_t^k = 1|s_r^k = 1)C{\hat x _{k|k - 1}})\]

Similarly in $W^{'}$, the term $E\{ {(s_t^k)^2}|s_r^k\}  = \mathbb{P}(s_t^k = 1|s_r^k)$. Using the same argument as above, we can write:
\begin{eqnarray}
E\{ {(s_t^k)^2}|s_r^k\}  - {({p_t})^2} &=& \mathbb{P}(s_t^k = 1|s_r^k = 1) - \mathbb{P}{(s_t^k = 1|s_r^k = 1)^2}\nonumber
\end{eqnarray}
including both cases. For convenience, we denote $p = \mathbb{P}(s_t^k = 1|s_r^k = 1)$ and finish the proof.
\end{proof}

In the next lemma we compute $p$ in the optimal linear estimator.

\begin{lemma}\label{lemma:probability}
Assume in the two-switch model that there are $h$ independent PUs, $\{ {u_1},...{u_h}\}$ in the sensing regions of ST only, and another $m$ independent PUs $\{ {u_{h + 1}},\cdots, {u_{h + m}}\}$ in the intersection of the sensing regions of ST and SR, and another $o$ independent PUs $\{ {u_{h + m + 1}},\cdots ,{u_{h + m + o}}\}$ in the receiver sensing region of SR only. Let $\{ {p_1},\cdots ,{p_{h + m + o}}\}$ denote the probabilities that the PUs are inactive, respectively. Then,
\begin{equation}
p = \mathbb{P}(s_t^k = 1|s_r^k = 1) = \mathop \Pi \limits_{i = 1}^h {p_i}
\label{eq:Probability}
\end{equation}
(\ref{eq:Probability}) means that the optimal linear estimator depends on the probabilities of inactive PUs that only exist in the sensing region of ST.
\end{lemma}
\begin{proof}
Note that the probabilities that ST can transmit and SR receive are:
\begin{eqnarray}
\mathbb{P}(s_t^k = 1) &=& \mathop \Pi \limits_{i = 1}^{h + m} {p_i}  \\
\mathbb{P}(s_r^k = 1) &=& \mathop \Pi \limits_{i = h + 1}^{h + m + o} {p_i}
\end{eqnarray}
respectively. It follows,
\begin{eqnarray}
\mathbb{P}(s_t^k = 1|s_r^k = 1) &=& \frac{{\mathbb{P}(s_t^k = 1,s_r^k = 1)}}{{\mathbb{P}(s_r^k = 1)}} \nonumber\\
&=& \frac{{\mathop \Pi \limits_{i = 1}^{h + m + o} {p_i}}}{{\mathop \Pi \limits_{i = h + 1}^{h + m + o} {p_i}}} = \mathop \Pi \limits_{i = 1}^h {p_i}\nonumber
\end{eqnarray}
proving the Lemma.
\end{proof}

\subsection{Example: Application to an Inverted Pendulum-Cart System}
To illustrate the performance of the optimal linear estimator, an application to estimate the states of an inverted pendulum-cart system via a CR system is proposed.

The parameters of the system to be estimated are given by:
\begin{eqnarray}
A &=& \left[ \begin{array}{l}
 {\rm{1}}{\rm{.0000 \ \ \ - 0}}{\rm{.0002 \ \ \ 0}}{\rm{.0010 \ \ \ - 0}}{\rm{.0000}} \\
 {\rm{0}}{\rm{.0000 \ \ \ 0}}{\rm{.9996 \ \ \ 0}}{\rm{.0001\ \ \ 0}}{\rm{.0010}} \\
 {\rm{0}}{\rm{.0315 \ \ \ - 0}}{\rm{.3901 \ \ \ 1}}{\rm{.0518 \ \ \- 0}}{\rm{.0417}} \\
 {\rm{0}}{\rm{.0726 \ \ \ - 0}}{\rm{.8763 \ \ \ 0}}{\rm{.1193 \ \ \ 0}}{\rm{.9038}} \\
 \end{array} \right] 
 \\
C &=& \left[ \begin{array}{l}
 {\rm{1 \ \ 0 \ \ 0 \ \ 0}} \\
 {\rm{0 \ \ 1 \ \ 0 \ \ 0}} \\
 \end{array} \right]
\\ 
W &=& \left[ \begin{array}{l}
 {\rm{0}}{\rm{.001   \ \  0}} \\
 {\rm{     0  \ \ 0}}{\rm{.001}} \\
 \end{array} \right] \\
V &=& \left[ \begin{array}{l}
 {\rm{0}}{\rm{.0100 \ \  0}}{\rm{.0090 \ \  0}}{\rm{.0020  \ \ 0}}{\rm{.0050}} \\
 {\rm{0}}{\rm{.0060 \ \  0}}{\rm{.0100  \ \  0}}{\rm{.0080  \ \ 0}}{\rm{.0060}} \\
 {\rm{0}}{\rm{.0040 \ \ 0}}{\rm{.0080  \ \ 0}}{\rm{.0030  \ \ 0}}{\rm{.0070}} \\
 {\rm{0}}{\rm{.0090 \ \  0}}{\rm{.0040 \ \  0}}{\rm{.0050  \ \ 0}}{\rm{.0100}} \\
 \end{array} \right]
\end{eqnarray}

Here $x = {[s;\,\,\theta ;\,\,v;\,\,\omega ]^T}$ is the state vector, with $s$ the position of the cart; $\theta$ the angle of the pendulum with the vertical line; $v$ the velocity of the cart; and $\omega$ the angular velocity of the pendulum.

The output signals are the position and angle of the inverted pendulum. The position should act as the reference signal at approximately 1m and the angle is set close to 0. The two switch model cognitive radio system in Fig. \ref{fig:twoswitch} is considered. Three PUs are detected in the sensing regions, assuming that ${p_1} = {p_2} = {p_3} = 0.8$. Using the proposed estimator we obtain the estimates of the position and angle shown in Fig. \ref{fig:estimationPer}. From the figure, it is obvious that the estimated states converge to the real ones.

\begin{figure}
 \begin{center}
  \includegraphics[width=8cm]{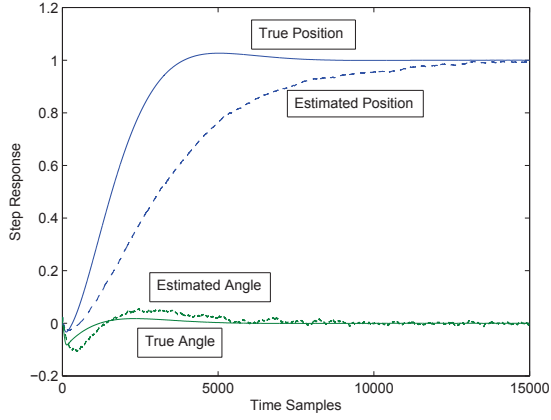}\\
   \end{center}
  \caption{State estimates of the position and the angle.}\label{fig:estimationPer}
\end{figure}

\subsection{Multi-channels Case}\label{sec:mul}

In the previous sections, we discussed one channel sensed a CR system, due to the independence of each channel, the results can be easily extended to the case of multiple channels. The optimal linear estimator is the same as derived in theorem \ref{thm:loe}. The difference is that the term $p = E\{ s_t^k|s_r^k\}$ needs to be recomputed. Assume that both the transmitter and the receiver choose the same channel to sense at each time. In this case, $s_t^k=s_{t}^{i_k}$ and $s_r^k=s_{r}^{i_k}$, where $i_k\in \{1,\cdots ,N\}$ represents the channel chosen to be sensed at time $k$. Thus, we have $p_k = \mathbb{P}\{ s_{t}^{i_k}=1|s_{r}^{i_k}=1\}$ following the previous discussion. This probability can be calculated as shown in lemma \ref{lemma:probability} for each $i_k\in\{1,\cdots ,N\}$.

If we consider how to choose the channel when the current channel is occupied by the PU, which means the correlation between channels could not be neglected. This requires a sensing strategy for the switch between channels, which would introduce more complexity to the probability analysis.
\begin{figure}
 \begin{center}
  \includegraphics[width=8cm]{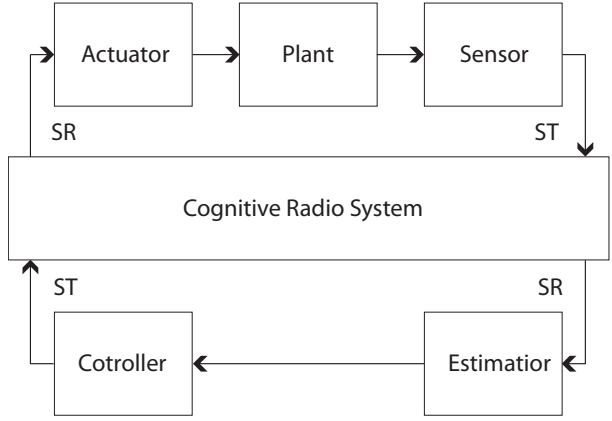}\\
   \end{center}
  \caption{Closed-Loop system.}\label{fig:closed}
\end{figure}
\section{Estimation and Control Through Cognitive Radio}\label{sec:ctrE}
\subsection{Estimation}
In this section we consider estimation and control of the closed-loop system when CR links exist between both the sensor to the estimator, and controller to the actuator, as shown in Fig. \ref{fig:closed}. We still focus on the case of one channel in the CR system, as the case of multiple channels can be addressed similarly as discussed in section 3.4. There are two STs, located at the sensor and controller ends, respectively, similarly for two SRs at estimator and actuator ends. Observe that the sensing variables are the same at the estimator and at the controller (the estimator and the controller are in the same location), thus for convenience we use $s_r^k$ to denote the sensing variable. Similarly, we use $s_t^k$ for the sensing variable of the actuator and sensor.

The system represented in Fig. \ref{fig:closed} becomes:
\begin{eqnarray}
{x_{k + 1}} &=& A{x_k} + Bs_t^k(s_r^k{u_k} + {\upsilon _k})\nonumber\\
{y_k} &=& s_r^k(s_t^kC{x_k} + {\omega _k})
\label{eq:closedPr}
\end{eqnarray}

The information set becomes $I_k=\{y_1,\cdots ,y_k,s_r^1,\cdots ,s_r^k,u_1,\cdots ,u_{k-1}\}$. The optimal linear estimator for this system is to minimize the cost function defined in section \ref{sec:pf} by assuming that the state estimate is a linear combination of measurements.

The a priori state estimate can be computed similarly as follows:
\begin{eqnarray}
{\hat x _{k + 1|k}} &=& A{\hat x _{k|k}} + ps_r^kB{u_k}\nonumber\\
{P_{k + 1|k}} &=& A{P_{k|k}}{A^T} + p(1 - p)s_r^kB{u_k}u_k^T{B^T} \nonumber\\
&&+ p_d BV{B^T}
\end{eqnarray}
where $p_d=p$ when $s_r^{k} = 1$ and $p_d=\mathbb{P}(s_t^k = 1|s_r^k = 0)$ when $s_r^{k} = 0$, the latter probability can be computed as in lemma \ref{lemma:probability} and $\mathbb{P}(s_t^k = 1|s_r^k = 0)=({\mathop \Pi \limits_{i = 1}^{n + m} {p_i}(1-\mathop \Pi \limits_{n + m + 1}^{n + m + o} {p_i})})/({1-\mathop \Pi \limits_{n + 1}^{n + m + o} {p_i}})$.

After receiving the measurement the a posteriori state estimate is obtained:
\begin{eqnarray}
{\hat x _{k + 1|k + 1}} &=& {\hat x _{k + 1|k}} + {K_{k + 1}}({y_{k + 1}} - s_r^{k + 1}pC{\hat x _{k + 1|k}})\nonumber\\
W_{k + 1}^{'} &=& \,\,\,W + (p - {(p)^2})C{X_{k + 1}}{C^T}\nonumber\\
{P_{k + 1|k + 1}} &=& {P_{k + 1|k}} - s_r^{k + 1}p{K_{k + 1}}C{P_{k + 1|k}}\nonumber\\
{K_{k + 1}} &=& {P_{k + 1|k}}p{C^T}{(pC{P_{k + 1|k}}p{C^T} + W_{k + 1}^{'})^{ - 1}}\nonumber\\
{X_{k + 1}} &=& \mathbb{E}\{ {x_{k + 1}}x_{k + 1}^T|I_{k+1}\}
\label{eq:EstimatorClose}
\end{eqnarray}

After additional computations from (\ref{eq:EstimatorClose}) we get:
\[\begin{array}{l}
 {X_{k + 1}} = \mathbb{E}\{ {x_{k + 1}}x_{k + 1}^T|{I_{k+1}}\}  
 = E\{ (A{\hat x _{k|k}} + A{e_{k|k}} + Bs_t^k(s_r^k{u_k})){( \cdots )^T}|{I_k}\}   + p_d BV{B^T} \\
  = (A{\hat x _{k|k}} + ps_r^kB{u_k}){( \cdots )^T} + p_d BV{B^T} + A{P_{k|k}}{A^T}  + (p - {p^2})B{u_k}u_k^T{B^T} \\
  = (A{\hat x _{k|k}} + ps_r^kB{u_k}){(A{\hat x _{k|k}} + ps_r^kB{u_k})^T} + {P_{k + 1|k}} \\
 \end{array}\]

In the next section, control design of the closed-loop system is introduced.

\subsection{Control}

It can be seen from the above that the error covariance is a function of the control input, this implies that the separation principle does not hold. An example is provided to illustrate that it is indeed the case. Assume a SISO system with $A=1$, $B=1$, $C=1$, $W=1$ and $V=0$. Consider the value function defined as:
\begin{eqnarray}
{V_N}({x_N}) &=& \mathbb{E}\{ {x_N}^T{Q_N}{x_N}|{I_N}\}\nonumber\\
{V_k}({x_k}) &=& \mathop {\min }\limits_{{u_k}} \mathbb{E}\{ {x_k}^T{Q_k}{x_k} + s_t^ks_r^k{u_k}^T{R_k}{u_k} + {V_{k + 1}}({x_{k + 1}})|{I_k}\}
\end{eqnarray}

Also assume $Q_N=Q_k=1$ and $R=0$.

When $k=N$, ${V_N}({x_N}) = \mathbb{E}\{ {x^2}_N|{I_N}\}$. When $k=N-1$,
\[\begin{array}{l}
 {V_{N - 1}}({x_{N - 1}}) = \mathop {\min }\limits_{{u_{N - 1}}} \mathbb{E}\{ ({x_{N - 1}}^2 + {V_N}({x_N}))|{I_{N - 1}}\}  \\
  = \mathbb{E}\{ (2{x_{N - 1}}^2)|{I_{N - 1}}\}  + \mathop {\min }\limits_{{u_{N - 1}}} \{ 2ps_r^{N - 1}{u_{N - 1}}{\hat x _{N - 1|N - 1}} + ps_r^{N - 1}{u_{N - 1}}^2\}  \\
 \end{array}\]

To compute the control action for step $k=N-1, N-2$, we only need to differentiate the above equation on both sides with respect to $u_{N-1}$. The steps are omitted and can be found in \cite{ACC:xiao}. We also give a detailed derivation of ${V_{N - 1}}({x_{N - 1}})$ to show that it is a function of the error covariance, and thus the separation principle does not hold.



The only case where the optimal controller is a linear gain of the state estimate is when $p=1$, which means that no PU exists in the transmitter sensing region. Using Fig. \ref{fig:twoswitch} as an example, when the PU A does not exist or exists in the intersection of the transmitter and receiver's sensing regions, there is an optimal controller that is a linear function of the state estimate. This provides us with an interesting insight: In order to obtain a linear optimal controller, the receiver should be located at a position where all PUs are covered by its sensing region.

\subsection{Stochastic Closed-Loop System Stability}

As seen above, the optimal controller depends on the estimation error covariance and is in fact a nonlinear function of the state estimate. It is therefore not obvious to study in detail the stability of the optimal closed-loop without an explicit expression of the controller. To simplify the problem we assume a suboptimal controller that is a linear function of the state estimate of the form ${u_k} =  - F{\hat x _{k|k - 1}}$, where $F$ is a constant matrix that is chosen such that $A-BF$ is stable. We are going to derive stability conditions of the corresponding closed-loop. Note in this case the error covariance is still a function of the control input. The state equations of the closed-loop system are derived as:
\begin{eqnarray}
{\hat x _{k + 1}} &=& (A - ps_r^kBF - pA{K_k}s_r^kC){\hat x _k} + A{K_k}s_r^ks_t^kC{x_k} + A{K_k}s_r^k{\omega _k}\nonumber\\
{x_{k + 1}} &=& A{x_k} - s_r^ks_t^kBF{\hat x _k} + Bs_t^k{\upsilon _k}
\end{eqnarray}
where ${\hat x _{k + 1}}\,: = {\hat x _{k + 1|k}}$.

Define
\begin{equation}
e_{k + 1}:=\varepsilon_{x,k+1|k}={x_{k+1}} - {\hat x _{k+1|k}}
\end{equation}
Subtracting the equations above and incorporating them in the closed-loop system:
\begin{eqnarray}
&&\left[ \begin{array}{l}
 {e_{k + 1}} \\
 {\hat{x}_{k + 1}} \\
 \end{array} \right] = \left[ \begin{array}{l}
 A - A{K_k}s_r^ks_t^kC \ \ \\
 A{K_k}s_r^ks_t^kC  \ \ \\
 \end{array} \right.\nonumber\\
 &&\left. \begin{array}{l}
 (A{K_k}C + BF)s_r^k(p - s_t^k) \\
 A - ps_r^kBF- A{K_k}Cs_r^k(p - s_t^k)\\
 \end{array} \right]
 \times \left[ \begin{array}{l}
 {e_k} \\
 {\hat{x}_k} \\
 \end{array} \right] \nonumber\\
 &&+ \left[ \begin{array}{l}
 Bs_t^k \ \  \\
 0 \ \ \\
 \end{array} \right.\left. \begin{array}{l}
  - A{K_k}s_r^k \\
 A{K_k}s_r^k \\
 \end{array} \right]\left[ \begin{array}{l}
 {\upsilon _k} \\
 {\omega _k} \\
 \end{array} \right]
 \label{eq:closed}
\end{eqnarray}

The conditions for the mean stability \cite{Auto:Koning} of the closed-loop system are given in the following Theorem.

\begin{theorem}\label{thm:meanstable}
The closed-loop system equation (\ref{eq:closed}) is m-stable (mean stable) if and only if the following conditions are satisfied:\\
(i) $|\rho (A - pqBF)| < 1$;\\
(ii) $|\rho (A - pqA{\tilde{K} _k}C)| < 1,\forall N$, for all $k \ge N$.\\
where $\rho (Z)$ represent the spectral radius of the matrix $Z$ and ${\tilde{K} _k} = \mathbb{E}\{ {K_k}\}$,where $K_k$ is a function of $\{ s_r^1,...,s_r^{k - 1}\}$ computed in section 3.2.
\end{theorem}
\begin{proof}
By definition of mean stability \cite{Auto:Koning}, taking the expectation of both sides of (\ref{eq:closed}) yields:
\begin{eqnarray}
\mathbb{E}\left\{ {\left[ \begin{array}{l}
 {e_{k + 1}} \\
 {\hat{x}_{k + 1}} \\
 \end{array} \right]} \right\} &=& \left[ \begin{array}{l}
 A - pqA{\tilde{K} _k}C\ \ \\
 pqA{\tilde{K} _k}C\ \ \\
 \end{array} \right.\left. \begin{array}{l}
 \mathop {\begin{array}{*{20}{c}}
   {} & 0  \\
\end{array}}\limits^{}  \\
 \mathop {}\limits^{} A - pqBF \\
 \end{array} \right]
 \times\mathbb{E}\left\{ {\left[ \begin{array}{l}
 {e_k} \\
 {\hat{x}_k} \\
 \end{array} \right]} \right\}
 \label{eq:expClosed}
\end{eqnarray}
where $\tilde{K} _k$ comes from $\mathbb{E}\{ {K_k}s_r^ks_t^k\}  = pq\mathbb{E}\{ {K_k}\}  = pq{\tilde{K} _k}$.

From (\ref{eq:expClosed}) the m-stability conditions (i) and (ii) of Theorem \ref{thm:meanstable} are established.
\end{proof}

\begin{remark}:
Condition (i) can be used as a necessary condition for the closed-loop stability for $F$ when $p$, $q$ are known. It provides a way to update the suboptimal linear controller for known $p$, $q$, with a new gain $\tilde{F}$ as to stabilize ${x_{k + 1}} = A{x_k} + pqB{u_k}$. We show in the numerical examples that this new gain will improve the performance of the closed-loop system. In condition (ii), $\{\tilde{K}_k\}_{k\geq 0}$ is a deterministic time varying sequence of matrices that can be computed as:
\begin{eqnarray}
\tilde{K} _k &=& q\mathbb{E}\{ {K_k}|\,s_r^{k - 1} = 1\}  + (1 - q)\mathbb{E}\{ {K_k}|\,s_r^{k - 1} = 0\} \nonumber\\
&=& q\mathbb{E}\{ \tilde{K}_k^1\} + (1 - q)\mathbb{E}\{ \tilde{K}_k^0\} \nonumber
\end{eqnarray}
where $\tilde{K_k^1}$ and $\tilde{K_k^0}$ are functions of $\{ s_r^1, \cdots ,s_r^{k - 2}\} $ and can be computed by letting $s_r^{k-1}=1$ and $s_r^{k-1}=0$ back into the estimator equations (\ref{eq:loeCR}), respectively. Similarly, $\mathbb{E}\{ \tilde{K}_k^1 \}  = q\mathbb{E}\{ \tilde{K}_k^1 |s_r^{k - 2} = 1\,\,\}  + (1 - q) \times \mathbb{E}\{ \tilde{K}_k^1 |s_r^{k - 2} = 0\}$, $\mathbb{E}\{ \tilde{K}_k^0 \}  = q\mathbb{E}\{ \tilde{K}_k^0 |\,s_r^{k - 2} = 1\}  + (1 - q) \times \mathbb{E}\{ \tilde{K}_k^0 |s_r^{k - 2} = 0\}$, and so on. Thus, applying the same reasoning from $s_r^{k-1}$ to $s_r^1$ and using equations (\ref{eq:loeCR}), $\tilde{K}_k$ is obtained.
\end{remark}

Next, we turn to a special but simplified case. This is the case when $p=1$, (\ref{eq:closedPr}) reduces to:
\begin{eqnarray}
{x_{k + 1}} &=& A{x_k} + Bs_r^k{u_k} + Bs_i^k{\upsilon _k}\nonumber\\
{y_k} &=& s_r^k(C{x_k} + {\omega _k})
\end{eqnarray}
where $s_i^k$ represents whether PUs in the intersection region of the ST and SR are active or not. The problem then becomes a packet loss problem that has been considered in \cite{IEEE:Schenato,NOLC:Imer}. However, the calculation of the optimal controller needs the exact value of $q$ which is hard to predict in CR systems as it is governed by the PUs' behavior. We give sufficient stability conditions of the so-called peak covariance process which can be viewed as an estimate of filtering deterioration caused by disruptions from PUs. First, we introduce the following definition then the conditions are provided in lemma \ref{lemma:Sufftwoswitch}.

\begin{definition}
Assume that $(A,B)$ is controllable, and $(A,C)$ is observable. The observability and controllability index are the smallest integer $I_0$ and $I_1$ such that $[{C'},{A'}{C'},\cdots ,{({A^{{I_0} - 1}})'}{C'}]$ and $[B,AB,\cdots ,({A^{{I_1} - 1}})B]$ have rank $n$, respectively \cite{Auto:Huang}.
\end{definition}

When $p=1$, (\ref{eq:closed}) becomes:
\begin{eqnarray}
\left[ \begin{array}{l}
 {e_{k + 1}} \\
 {\hat{x}_{k + 1}} \\
 \end{array} \right] &=& \left[ \begin{array}{l}
 A - A{K_k}s_r^kC\ \ \\
 A{K_k}s_r^kC\ \ \\
 \end{array} \right.\left. \begin{array}{l}
 0 \\
 A - s_r^kBF \\
 \end{array} \right]\nonumber\\
 \times\left[ \begin{array}{l}
 {e_k} \\
 {\hat{x}_k} \\
 \end{array} \right]
 &+& \left[ \begin{array}{l}
 Bs_i^k\ \   \\
 0\ \  \\
 \end{array} \right.\left. \begin{array}{l}
  - A{K_k}s_r^k \\
 A{K_k}s_r^k \\
 \end{array} \right]\left[ \begin{array}{l}
 {\upsilon _k} \\
 {\omega _k} \\
 \end{array} \right]
 \end{eqnarray}

Let ${L_{k + 1}} = \mathbb{E}\left\{ \left[ \begin{array}{l}
 {e_{k + 1}} \\
 {\hat{x}_{k + 1}} \\
 \end{array} \right]{\left[ \begin{array}{l}
 {e_{k + 1}} \\
 {\hat{x}_{k + 1}} \\
 \end{array} \right]^T}|{I_k}\right\}$. Assume the initial condition $s_r^1=1$. The following two stopping times are introduced \cite{Auto:Huang}:
\begin{center}
$\alpha_1= \inf\{k:k>1,s_r^k=0\}$.\\
$\beta_1 = \inf\{k:k>\alpha_1,s_r^k=1\}$.
\end{center}

Thus $\alpha_1$ is the first time when primary users occur and $\beta_1$ is the first time the channel becomes idle again. The above procedure then generates two sequences:
\begin{center}
$\alpha_1,\ \alpha_2,...,\ \alpha_n,...$\\
$\beta_1,\ \beta_2,...,\ \beta_n,...$
\end{center}

where for $j>1$:
\begin{center}
$\alpha_j=inf\{k:k>\beta_{j-1},s_r^k=0\}$.\\
$\beta_j=inf\{k:k>\alpha_j,s_r^k=1\}$.
\end{center}

Denote $L_{n}^p=L_{\beta_n}$, $\{L_{n}^p\}_{n\geq 1}$ is called the peak covariance process (also a subsequence process) of $\{L_{k}\}_{k\geq 1}$ \cite{Auto:Huang}. The peak covariance is computed at the last time instant of a consecutive $s_r^k=0$. Its stability analysis is important and useful for analyzing system performance, in that it provides an insight in how "bad" the covariance process might be due to successive packet losses.

\begin{definition}\cite{Auto:Huang}
We say the peak covariance sequence $\{L_{n}^p\}_{n\geq 1}$ is stable if $\sup_{n\geq{1}}{\mathbb{E}\parallel{L_{n}^p}\parallel}<\infty$. Accordingly, we say the system satisfies peak covariance stability.
\end{definition}

\begin{lemma} \label{lemma:Sufftwoswitch}
$\{L_{n}^p\}_{n\geq 1}$ is stable if the following two conditions hold:
\begin{eqnarray}
(i)&&q \ge 1 - \frac{1}{{{{\max }_i}|{\lambda _i}(A){|^2}}}\nonumber\\
(ii)&&(1 - q)qd_1^{(1)}[1 + \sum_{i = 1}^{I - 1}{d_i^{(1)}{q^i}} ]\sum_{j = 1}^\infty \parallel{A^j}\parallel{^2}\nonumber\\
&&\times {(1 - q)^{j - 1}} < 1\nonumber
\end{eqnarray}
where $\lambda_A$ is an eigenvalue of the largest magnitude for the matrix $A$, $I = \max \{ {I_0},{I_1}\}$ and $d_i^{(1)}$ is a positive constant given in the proof below.
\end{lemma}
\begin{proof}
Since the separation principle holds in this case, estimation and control can be performed separately. Consider the error covariance
\begin{eqnarray}
P_{k+1}=\mathbb{E}\{e_{k+1}e_{k+1}^{T}|I_k\}=AP_k A^T+Bp_j VB^T-s_r^kAP_kC^T(CP_kC^T+W)^{-1}CP_kA^T
\label{eq:estimation}
\end{eqnarray}
where $P_{k+1}=P_{k+1|k}$ and $p_j=\mathbb{E}\{s_i^k|s_r^k\}$.

Define $F(P) = APA^T+p_j BVB^T-APC^T(CPC^T+W)^{-1}CPA^T$. When $1 \le i \le \max (I_0,1)$, there always exist $c_i^{(1)}\geq 0$ and $c_i^{(0)}$ that satisfy the following inequality  \cite{Auto:Huang}:
\[\parallel {F^i}(P)\parallel\,\, \le \,\,c_i^{(1)}\parallel P\parallel + c_i^{(0)}\]
where $\parallel X\parallel$ refers to the matrix induced norm $\parallel X\parallel=\max_{|X|=1}|MX|$ where $|X|$ and $|MX|$ denote the usual Euclidean norm for vectors.

By \cite{Auto:Huang}, where only estimation is considered, the peak covariance process $\{P_{n}^p\}_{n\geq 1}$ of $\{P_k\}_{k\geq 1}$ is stable if condition (i) above holds and
\begin{equation}
(1 - q)qc_1^{(1)}[1 + \sum_{i = 1}^{{I_0} - 1}{c_i^{(1)}{q^i}} ]\sum _{j = 1}^\infty \parallel{A^j}\parallel{^2}{(1 - q)^{j - 1}} < 1
\label{eq:huang}
\end{equation}
is satisfied.

Consider the control part, from the close-loop equation above, we have
\begin{equation}
\hat{x}_{k+1} = s_r^kAK_kCe_k+(A-s_r^kBF)\hat{x}_k+s_r^kAK_kw_k
\end{equation}
Let $M_{k+1}=\mathbb{E}\{\hat{x}_{k+1}\hat{x}_{k+1}^T|I_k\}$, then we have
\begin{eqnarray}
M_{k+1} = AM_kA^T+T_k-s_r^k(BFM_kA^T +AM_kF^TB^T-BFM_kF^TB^T)
\label{eq:control}
\end{eqnarray}
where $T_k=s_r^kAP_kC^T(CP_kC^T+W)^{-1}CP_kA^T$. $\parallel T_k\parallel$ is bounded if condition (i) and (\ref{eq:huang}) hold. To see this, note $\{P_{n}^p\}_{n\geq 1}$ is stable, then from each $\beta_n$ to $\alpha_{n+1}$, $s_r^k=1$ for a successive period, it follows that $P_k$ is bounded based on Kalman filtering theory. This implies that $\parallel T_k\parallel$ is also bounded in that period. Once $s_r^k$ becomes $0$, $T_k=0$ in that period.

Define $G(M) = AM A^T+T_s-(BFMA^T+AMF^TB^T-BFMF^TB^T)$, where $T_s:=\{T_k:\ \parallel T_k\parallel=sup_{g\geq 1}\parallel T_g\parallel\}$. Similarly, when $1 \le i \le \max (I_1,1)$, there always exist $e_i^{(1)}\geq 0$ and $e_i^{(0)}$ that satisfy the following inequality  \cite{Auto:Huang}:
\[\parallel{G^i}(M)\parallel\,\, \le \,\,e_i^{(1)}\parallel M\parallel + e_i^{(0)}\]

Following the same arguments in \cite{Auto:Huang}, besides condition (i) and (\ref{eq:huang}),
\begin{equation}
(1 - q)qe_1^{(1)}[1 + \sum_{i = 1}^{{I_1} - 1}{e_i^{(1)}{q^i}} ]\sum_{j = 1}^\infty \parallel{A^j}\parallel{^2}{(1 - q)^{j - 1}} < 1
\label{eq:huangcontrol}
\end{equation}
is also satisfied (note in (\ref{eq:huangcontrol}) $I_1$ is used instead of $I_0$ in (\ref{eq:huang})), then the peak covariance process $\{M_{n}^p\}_{n\geq 1}$ of $\{M_k\}_{k\geq 1}$ is stable. We have now
\begin{eqnarray}
{L_{k + 1}} &=& \mathbb{E}\left\{ \left[ \begin{array}{l}
 {e_{k + 1}} \\
 {\hat{x}_{k + 1}} \\
 \end{array} \right]{\left[ \begin{array}{l}
 {e_{k + 1}} \\
 {\hat{x}_{k + 1}} \\
 \end{array} \right]^T}|{I_k}\right\}\nonumber\\
 &=&\left[ \begin{array}{cccc}
 {P_{k + 1}} & 0\\
 0 & {M_{k + 1}}\\
 \end{array} \right]\nonumber
 \end{eqnarray}

Thus, $L_{n}^p=\left[ \begin{array}{cccc}
 {P_{n}^p} & 0\\
 0 & {M_{n}^p}\\
 \end{array} \right]$ is stable. Define for $1\leq i\leq \min\{(I_0-1),(I_1-1)\}$, $d_i^{(1)}=max\{c_i^{(1)},e_i^{(1)}\}$; For $i>\min\{(I_0-1),(I_1-1)\}$, if $I_0\geq I_1$, $d_i^{(1)}=c_i^{(1)}$, otherwise $d_i^{(1)}=e_i^{(1)}$. Then, we can combine (\ref{eq:huang}) and (\ref{eq:huangcontrol}) together and get condition (ii) in the lemma.
\end{proof}

\begin{remark}:
Lemma \ref{lemma:Sufftwoswitch} gives sufficient conditions for a linear gain to stabilize the CR closed-loop system when the optimal controller cannot be obtained.
\end{remark}

In the next section, illustrative examples are provided to show improved performance of the closed-loop system through CR links and test the stability conditions.

\subsection{Numerical Examples}

We consider the model of the CR system shown in Fig. \ref{fig:twoMath}, with instable inverted pendulum-cart system parameters:\\
$A = \left[ \begin{array}{l}
 {\rm{1}}{\rm{.0000\ \  0}}{\rm{.0000 \ \ 0}}{\rm{.0010 \ \ - 0}}{\rm{.0000}} \\
 {\rm{0}}{\rm{.0000 \ \ 1}}{\rm{.0000\ \  - 0}}{\rm{.0000 \ \  0}}{\rm{.0010}} \\
 {\rm{0}}{\rm{.0000\ \  0}}{\rm{.0022 \ \ 0}}{\rm{.9842 \ \ - 0}}{\rm{.0000}} \\
 {\rm{0}}{\rm{.0000 \ \ 0}}{\rm{.0278\ \  - 0}}{\rm{.0363 \ \ 0}}{\rm{.9999}} \\
 \end{array} \right]$,
${\rm{B  = }}{\left[ {{\rm{0}}{\rm{.0000\ \ ,0}}{\rm{.0000\ \ ,0}}{\rm{.0023\ \ ,0}}{\rm{.0052}}} \right]^T}$,\\
\\
When the three PUs are detected in the sensing regions, assume that ${p_1} = {p_2} = {p_3} = 0.8$. The controller is a suboptimal LQR to the linear deterministic system. We use the LQR gain for the deterministic system ${x_{k + 1}} = A{x_k} + B{u_k}$ where ${u_k} =  - F{x_k}$:
\[F = [{\rm{ - 13}}{\rm{.9382 \ \ 173}}{\rm{.6752  \ \ - 29}}{\rm{.9030\ \ 18}}{\rm{.4750}}]\]

We can see the step response is the desired result in Fig. \ref{fig:CLdesired}. Then when we fix $F$ and change $p_1=0.5$, we can see that the step response diverges in Fig. \ref{fig:CLunstable}.
\begin{figure}
 \begin{center}
  \includegraphics[width=8cm]{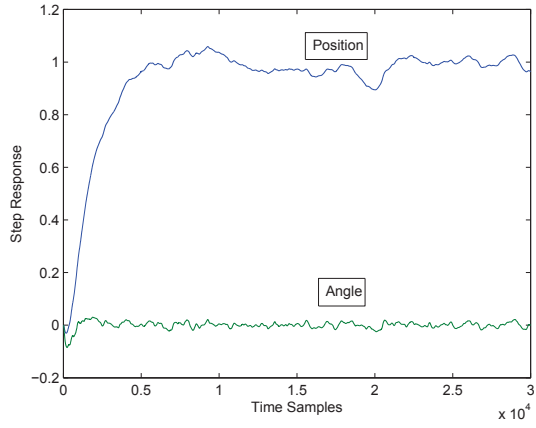}\\
   \end{center}
  \caption{Step response of the closed-loop system.}\label{fig:CLdesired}
\end{figure}
\begin{figure}
 \begin{center}
  \includegraphics[width=8cm]{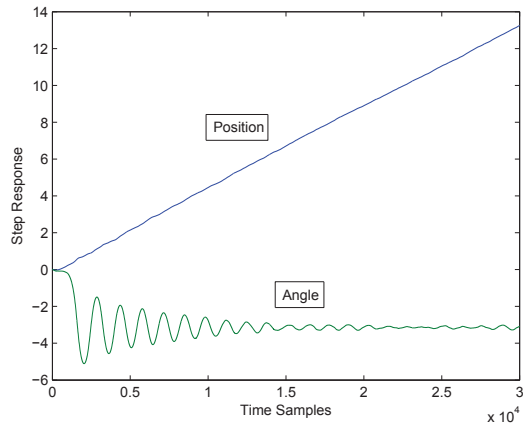}\\
   \end{center}
  \caption{Step response with more activity of PUs.}\label{fig:CLunstable}
\end{figure}

\begin{figure}
 \begin{center}
  \includegraphics[width=8cm]{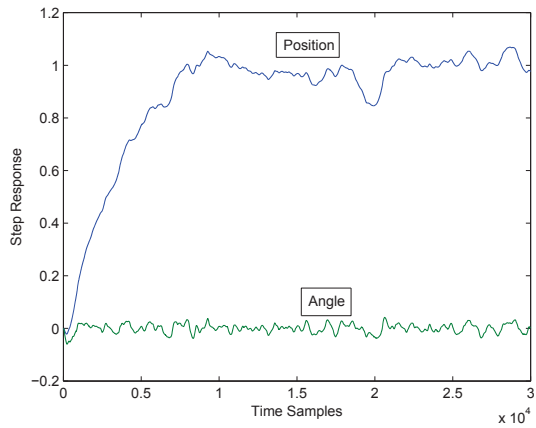}\\
   \end{center}
  \caption{Step response for a better controller gain.}\label{fig:CLupdated}
\end{figure}
\begin{figure}
 \begin{center}
  \includegraphics[width=8cm]{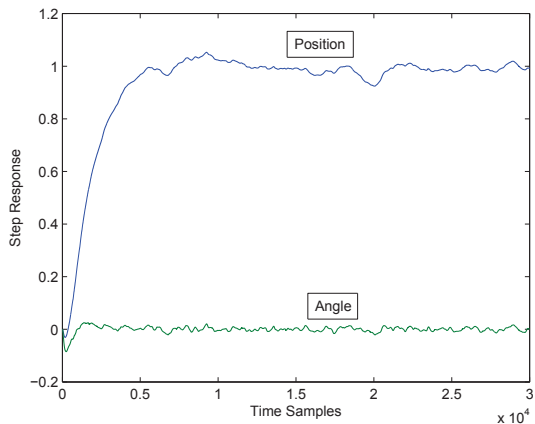}\\
   \end{center}
  \caption{Step response when $p=1$.}\label{fig:CLtest}
\end{figure}

Next, we set ${p_1} = 0.5,{p_2} = {p_3} = 0.8$, but design a LQR for the system ${x_{k + 1}} = A{x_k} + pqB{u_k}$ as suggested, and run the step response for the closed-loop system, which produces Fig. \ref{fig:CLupdated} shows that for this improved design the step response of the system is stable.

Finaly, set ${p_1} = 1,{p_2} = 0.7,{p_3} = 0.8$. Here we have $I = I_1=I_0= 4$ and $||{F^i}(P)||\, \le \,||{A^i}{A^i}^T||||P||$ and $||{G^i}(M)||\, \le \,||{A^i}{A^i}^T||||M||$. Thus we take $d_1^{(1)} = 1.0395$, $d_2^{(1)} = 1.0805$, $d_3^{(1)} = 1.1230$. After some computations, the left-hand side of condition (ii) in lemma \ref{lemma:Sufftwoswitch} is approximately $0.9997\le 1$ and the condition is satisfied. The result is depicted in Fig. \ref{fig:CLtest}. The step response is stable, so the error covariance and thus the peak covariance are both stable. Note that in this case the optimal controller exists and is linear in the state estimates.

\section{Conclusion}\label{sec:con}
In this paper, estimation and control of CR systems is investigated. A CR system model, called the two-switch model, proposed in the communication community to represent the communication CR links is employed. The optimal linear estimator is derived for a single CR link. Then estimation and control of the closed-loop system over two CR links are addressed. The controller is shown to be a non-linear function of the state estimates. Several stability conditions are derived and numerical examples are presented to illustrate the effectiveness of results. As a new emerging communication technology, CR combined with control theory may pave the way for new applications, for example, see \cite{TWE:xiao}. The future work will focus on more practical transmission models of CR systems which are generated from the sensing spectrum model and the corresponding estimation and control algorithms, and also the probability analysis on the stability for the case $p_1\neq 1$ and more PUs involved. Moreover, less conservative stability conditions will be derived based on \cite{Auto:you}.

\section{Acknowledgement}\label{sec:ack}
This work was supported in part by NSF grant NSF-CMMI 1334094. This paper has been authored by employees of UT-Battelle, LLC, under contract DE-AC05-00OR22725 with the U.S. Department of Energy. Accordingly, the United States Government retains and the publisher, by accepting the article for publication, acknowledges that the United States Government retains a non-exclusive, paid-up, irrevocable, world-wide license to publish or reproduce the published form of this manuscript, or allow others to do so, for United States Government purposes.

\end{document}